\documentclass[aps,pra,twocolumn,showkeys,showpacs,groupedaddress]{revtex4}
\usepackage{tikz}
\usetikzlibrary{positioning}
\usepackage{amsmath,paralist,amsthm,comment,amssymb}

\newcommand{\ket}[1]{|#1\rangle}

\newcommand{\mc}[1]{\mathcal{#1}}

\newcommand{\F}{\mathbb{F}}

\newcommand{\lu}{\rm{LU}}
\newcommand{\lc}{\rm{LC}}
\newcommand{\supp}{{\rm{supp}}}
\newcommand{\wt}{{\rm{wt}}}
\newcommand{\Tr}{{\rm{Tr}}}

\newcommand{\nix}[1]{}
\newcommand{\spn}[1]{{\left\langle #1\right\rangle}}

\newcommand{\be}{\begin{eqnarray*}}
\newcommand{\ee}{\end{eqnarray*}}
\newcommand{\ben}{\begin{eqnarray}}
\newcommand{\een}{\end{eqnarray}}

\newcommand{\ba}{\begin{array}}
\newcommand{\ea}{\end{array}}

\newtheorem{theorem}{Theorem}
\newtheorem{lemma}[theorem]{Lemma}
\newtheorem{corollary}[theorem]{Corollary}

\begin{document}

\title{On Local Equivalence, Surface Code States and Matroids}

\author{Pradeep Sarvepalli}
\email[]{pradeep@phas.ubc.ca}
\affiliation{Department of Physics and Astronomy,
University of British Columbia, Vancouver V6T 1Z1, Canada }
\author{Robert Raussendorf}
\affiliation{Department of Physics and Astronomy,
University of British Columbia, Vancouver V6T 1Z1, Canada }

\date{April 5, 2010}

\begin{abstract}
Recently,  Ji et al disproved the LU-LC  conjecture and showed that the local unitary and local Clifford equivalence classes of the stabilizer states are not always the same. Despite the fact this settles the LU-LC conjecture, a sufficient condition for stabilizer states that violate the LU-LC conjecture is missing. In this paper, we investigate further the properties of stabilizer states with respect to local equivalence. Our first result shows that there exist infinitely many stabilizer states which violate the LU-LC conjecture. 
In particular, we show that for all numbers of qubits $n\geq 28$, there exist distance two stabilizer states which are counterexamples to the  LU-LC conjecture. We prove that for all odd $n\geq 195$, there exist stabilizer states with distance greater than two which are LU equivalent but not LC equivalent.

Two important classes of stabilizer states that are of great interest in quantum computation are the cluster states and stabilizer states of the surface codes. To date, the status of these states with respect to the LU-LC
conjecture was not studied.  We show that, under some minimal restrictions, both these classes of states preclude any counterexamples. In this context, we also show that the associated surface codes do not have any encoded non-Clifford transversal gates. We characterize the CSS surface code states in terms of a class of minor closed binary matroids. In addition to making connection with an important open problem in binary matroid theory, this characterization does in some cases provide an efficient test for CSS states that are not counterexamples.
\end{abstract}


\maketitle
\section{Introduction} 
An important problem in quantum information theory is to understand and classify the equivalence classes of quantum states.
Such a classification could potentially  simplify certain tasks in processing of quantum information and improve our  understanding of the role of entanglement in quantum computing, especially, in measurement based computation models. A special case of the problem   that has attracted much  attention in the quantum information community, is to classify the equivalence of stabilizer states under local unitary gates and local Clifford gates.

Given a pair of stabilizer states, it is not
known whether there exists an efficient (i.e. polynomial time) algorithm to test if they are local unitary equivalent. However, local {\em{Clifford}} equivalence of a pair of stabilizer states can be efficiently tested \cite{vanDenNest04}. With this fact in mind, it has been conjectured that two stabilizer states are local unitary equivalent if and only if they are local Clifford equivalent. This is the LU-LC conjecture \cite{krueger05}. If it were true, then the above test \cite{vanDenNest04} for local Clifford equivalence would imply an efficient test for local unitary equivalence among stabilizer states.

The LU-LC conjecture has been verified in numerous special cases. Rains \cite{rains99c} showed that the automorphisms of the linear stabilizer codes, with distance greater than two,  are all in the local Clifford group. This result was subsequently strengthened to include a larger class of stabilizer states by
Van den Nest et al \cite{vanDenNest05} and Zeng et al \cite{zeng07}. Hein et al \cite{hein04} classified the stabilizer states of up to seven qubits, and the LU-LC conjecture always holds in this domain. These findings gave increasing credence to the  LU-LC conjecture.

However, as has recently been demonstrated by a 27-qubit counterexample, the LU-LC conjecture is false \cite{ji08}. Due to the above motivation, this surprising result may easily be seen as a setback. In a more positive light, it also reminds us how intricate a place Hilbert space is, often defying our preconceptions.

LU-LC counterexamples now come into view as an object of study in their own right. Presently, only a small number of computer-generated LU-LC counterexamples are known, and we lack an understanding of them. A theory of LU-LC counterexamples, once established, should identify the property due to which stabilizer states fail to satisfy the LU-LC conjecture, and provide a classification of all LU-LC counterexamples.

In this work we begin to study counterexamples to the LU-LC conjecture in a systematic fashion. We ask two questions: (1) {\em{Are there finitely or infinitely many LU-LC counterexamples?}}, and (2) {\em{Can LU-LC counterexamples be found among prominent families of stabilizer states?}} 
Regarding the first question, to date only very few stabilizer states are known that violate the LU-LC conjecture. One may therefore suspect they are rare and even finite in number. In Section~\ref{sec:counterEx} of this paper we show that there exist infinitely many counterexamples to the LU-LC conjecture.

To address the second question, in Sections~\ref{sec:cluster}~and~\ref{sec:surfaceStates} we shift our focus to two important classes of stabilizer states, the cluster states \cite{raussen01} and the surface code states \cite{kitaev03,bravyi98}. The cluster states form a universal resource for measurement based quantum computation, and the surface codes are an important family of quantum codes with topological properties. We show that under minimal restrictions on the states, they cannot be counterexamples to the LU-LC conjecture.

These results have useful applications. Quantum codes with encoded non-Clifford transversal gates are much sought after in the context of fault tolerant quantum computation. Surface codes are  of great interest because they are especially suited for fault tolerant computation and their
potential for high thresholds. However, no non-Clifford transversal gates have
been found for these codes since their discovery almost a decade ago. As a consequence of our results on LU-LC counterexamples, in Section~\ref{sec:surfaceStates} we show that under similar restrictions as for LU-LC equivalence, surface 
 codes do not have any non-Clifford transversal gates. Hence, fault tolerant quantum computation with these codes must necessarily rely on other methods to realize any encoded non-Clifford gates.

\section{Background}\label{sec:bg}
Let  $\mc{P}_n$ be the Pauli group over $n$ qubits. 
Let $g=\otimes_{i=1}^n g_i  $ be an element in $ \mc{P}_n$. Then the support of $g$ is defined
as the subset of $\{ 1,\ldots,n\}$ for which $g_i\neq I$ i.e.
\ben
\supp(g) = \{ i\mid g_i \neq I\}.
\een
Given a stabilizer state $\ket{\psi}$ and its stabilizer $S(\psi)$, we say that $g\in S(\psi)$ is a minimal support element
if there does not exist any element $h\in S(\psi)$ such that $\emptyset\neq\supp(h) \subset \supp(g)$. In other words, 
$\supp(g)$ does not strictly contain the support of any nontrivial element of $S$. We also say that $g$ is a 
minimal element of $S(\psi)$, and $\supp(g)$ is a minimal support of $S(\psi)$.
We define the weight of  $g \in \mc{P}_n$
as 
\ben
\wt(g) =|\{ i\mid g_i\neq I\}| = |\supp(g)|.
\een
The distance of a subgroup  $S$ in $\mc{P}_n$ is defined as 
\ben
\min_{g\in S\setminus I } \wt(g),
\een
where $g$ is not the identity element of $S$.   Often, we refer to the distance of a stabilizer state by which we mean the distance of its stabilizer $S(\psi)$.

Let $\ket{\psi}$ be a stabilizer state. Let $U(2^n)$ denote the group of $2^n\times 2^n$ unitary matrices. An element of $U(2^n)$
is local unitary  if it is in the local unitary group $\mc{U}_n^l = U(2)^{\otimes^n}$. Two stabilizer states $\ket{\psi}$ and $\ket{\psi'}$ are said to be local unitary ($\lu$)  equivalent if $\ket{\psi'}=U \ket{\psi}$  for some $U\in \mc{U}_n^l $. If each $U_i$ is also a diagonal matrix, then $\ket{\psi}$ and $\ket{\psi'}$ are  diagonal local
 unitary (DLU) equivalent. We denote the local unitary equivalence class of a stabilizer state $\ket{\psi}$ by $\lu(\psi)$.
The Clifford group over $n$ qubits is the normalizer of the Pauli group in $U(2^n)$. In other words, 
\ben
\mc{K}_n= \{ U \in U(2^n) \mid U \mc{P}_n U^\dagger = \mc{P}_n\}.\label{eq:cliffordGrp}
\een
The local Clifford group over $n$ qubits is defined as $\mc{K}_n^l = \mc{K}_1^{\otimes^n}$.  Two stabilizer states $\ket{\psi}$ 
and $\ket{\psi'}$ are local Clifford ($\lc$) equivalent if there exists a $K \in \mc{K}_n^l$ such that $\ket{\psi'} = K \ket{\psi}$.
We denote the local Clifford equivalence class of $\ket{\psi}$ by  $\lc(\psi)$. The single qubit Clifford group $\mc{K}_1 =\spn{\lambda I, H, P}$ where  $\lambda$ is a complex scalar, while $H$ and $P$ denote the  Hadamard gate and the phase gate ($\text{diag}(1,i)$). We consider two elements of $\mc{K}_1$ equivalent if their action on the 
Pauli matrices (by conjugation) differs only by a scalar. Then we need only consider the action of six elements\footnote{Up to a scalar.} in $\mc{K}_1$, 
namely---$\{I, H, P, HP, PH, HPH \}$. These can be identified with nonsingular $2\times 2$ binary matrices over the binary field 
$\F_2$.

The action of local unitaries on the stabilizer provides an important handle in understanding the
LU-LC equivalence classes of that state. In particular, the action of local unitaries on the 
minimal supports of a stabilizer state is of great significance. We state the relevant result
below and refer the interested reader to \cite{rains99c} for further details. 
\begin{lemma}\label{lm:min2min}
If $U\in \mc{U}_n^l$ maps a stabilizer state $\ket{\psi}$ to another stabilizer state $\ket{\psi'} = U\ket{\psi}$, then $U$ maps the minimal support elements of $S(\psi)$ to 
the minimal support elements of $S(\psi')$. 
\end{lemma}
If both the local equivalence classes of $\ket{\psi}$ are the same, then we indicate this by $\lu(\psi)=\lc(\psi)$. 
There are some conditions under which we can conclude that $\lu(\psi)=\lc(\psi)$.  A sufficient condition due to van den Nest et al  will be useful in this context.  
For proof and further details, please refer to \cite[Theorem~1]{vanDenNest05} and \cite[Corollary~1]{vanDenNest05}.
\begin{lemma}[Minimal Support Condition]\label{lm:msc}
Suppose that $\ket{\psi}$ is a  $n$-qubit  stabilizer state free from Bell pairs. Let its  stabilizer be $S(\psi)$ and  $M(\psi)$ be the group generated 
by  the minimal support elements of  $S(\psi)$.  If  all the Pauli matrices $X, Y, Z$  occur on every qubit  in $M(\psi)$, 
(i.e. for any $\alpha\in\{ X,Y,Z\}$,  and for any $i$, there exists some $\otimes_{j=1}^{n} g_j\in M(\psi)$ such that it $g_i=\alpha$),
then $\lu({\psi}) = \lc({\psi})$. In particular this condition holds  if $S(\psi) = M(\psi)$. 
\end{lemma}

\paragraph*{Remark:}
This is a slightly relaxed version of the result in \cite[Theorem~1]{vanDenNest05}, instead of fully entangled stabilizer states
we have stabilizer states that are free from Bell pairs. It can be shown that the original proof holds with this modified condition.

\subsection{Surface Code States}  

Let $\Gamma$ be a graph with vertex set $V(\Gamma)$ and edge set $E(\Gamma)$. 
A cycle is an alternating sequence of vertices and edges such that every edge connects
the adjacent vertices and only the first and last vertices are same. 
The length of a cycle is the number of edges in the cycle and a cycle of length $n$ is called an $n$-cycle.
A loop is an edge connecting a vertex to itself, i.e., it is a $1$-cycle. Not every graph can be drawn on the sphere 
so that the edges do not cross. However, by adding handles to the sphere we can draw the graph on the resulting 
surface so that the edges do not
cross. The genus of the graph is the smallest number of handles that we need to add to the sphere 
so that it can embedded without edge crossings. 
Let the genus of the graph be $g$ and assume that it is embedded on a surface of genus $g'\geq g$.
We denote the set of faces of $\Gamma$ by $F(\Gamma)$. The union of all the faces equals the surface on which the graph is embedded\footnote{Topologically, the surface is the union of the faces and the graph (the vertices and the edges) 
since the edges and the vertices do not belong to the faces. }. 
We can associate a stabilizer code to $\Gamma$. We identify qubits with  the edges of the graph.
Define the site operators and the face operators of $\Gamma$ as:
\ben
A_v =\prod_{e\in \delta(v) } X_e ;\quad B_f = \prod_{e\in\partial(f)} Z_e,
\een
where $\delta(v)$ is the set of edges incident on the vertex $v$ and $\partial(f)$ is the set of edges that constitute the boundary
of the face $f$.   Let $S= \spn{ A_v, B_f\mid v\in V(\Gamma), f\in F(\Gamma)} $, i.e, the group generated by the site operators and the face operators. We call the states stabilized by $S$ as the surface code states of $\Gamma$ and denote a surface code state of $\Gamma$ 
as $\ket{\psi_{\Gamma}}$ to distinguish it from the graph state
associated to $\Gamma$, where the vertices rather than the edges are associated to qubits.

\textbf{Dual Graphs.} Given a graph $\Gamma$, that is embedded on a surface, we can define a dual graph, denoted $\Gamma^\ast$. The dual graph is obtained as follows:
\begin{compactenum}[i)]
\item  Each face $f$ of $\Gamma$ corresponds to a vertex $f^\ast$ in $\Gamma^\ast$.
\item  Every vertex $v$ of $\Gamma$ corresponds to a face $v^\ast$ in $\Gamma^\ast$.
\item  For each edge in $\Gamma$ that is in the boundary of two faces $f_1,f_2$, 
connect the vertices $f_1^\ast$ and $f_2^\ast$ in $\Gamma^\ast$.
In other words, if $e\in \partial(f_1)\cap \partial(f_2)$, then we place an edge between
the vertices  $f_1^\ast$ and $f_2^\ast$ in $\Gamma^\ast$.
\end{compactenum}
We have the following correspondence between the $\Gamma$ and $\Gamma^\ast$. 
\begin{center}
\begin{tabular}{c|c}
$\Gamma$ &  $\Gamma^\ast$     \\
	\hline
	Edges & Edges\\
	Faces & Vertices\\
	Vertices & Faces\\
\end{tabular}
\end{center}

The edges in $\Gamma$ are in 1-1 correspondence with the edges in $\Gamma^\ast$ while the vertices are in correspondence with the faces
of the dual graph. 
The edges incident on a vertex in $\Gamma$ ($\Gamma^\ast$) are precisely the edges that form the boundary of the associated face in $\Gamma^\ast$ ($\Gamma$). 
So the operator $A_v$ is associated to a vertex and its incident edges in $\Gamma$, but it is associated to a cycle and its boundary edges in $\Gamma^\ast$. A set of edges is called an elementary cycle if it forms the boundary of a face in $\Gamma$.
A set of edges of $\Gamma$ which form a cycle in $\Gamma^\ast$ is called a cocyle of $\Gamma$.  A cocyle is called an
elementary cocycle if it forms the boundary of a face in $\Gamma^\ast$.  
A loop is  an edge connected to the same vertex. 
A coloop is a loop in the dual graph. 

If $\Gamma$ has no loops and has $n$ edges and is embedded on  a surface of genus $g$, the group generated by the site operators and the face operators contains $n-2g$ generators and it defines an $[[n,2g]]$ surface code. There are $2^{2g}$ surface code states associated with an $[[n,2g]]$  surface code.  In this case, the supports of the encoded operators are defined by the nontrivial cycles and cocycles of $\Gamma$. For instance,  when a graph of genus 2 is embedded on the torus we have an $[[n,2]]$ surface code with $4$ encoded operators; in general for a surface code on a surface of genus 
$g$ we have $2g$ encoded operators. The support of the encoded 
$Z$ operators is given by the cycles in $\Gamma$ that wind across the holes of the torus and  while that of encoded $X$ operators are the cycles that wind around the holes of the torus in the dual graph. In other words the encoded operators are cycles and cocycles of $\Gamma$ that are not homologous to any of  the elementary cycles or cocycles. 

When the graph has loops or when its genus is smaller than the genus of the
surface on which it has been embedded, the number of encoded qubits can vary. 
For a graph with $n$ edges and genus $g$ that is embedded on a surface of genus $g'\geq g $, the stabilizer has
 $n-k$ generators and $2k$ encoded operators, where $k\leq 2g'$.
Denote the encoded operators as $\mc{L}=\{\overline{X}_1, \overline{Z}_1, \ldots, \overline{X}_k, \overline{Z}_k \}$. 
The support of any encoded operators is a  nontrivial cycle of the surface.
These encoded operators allow us to specify each surface code state as follows. Let $\mathfrak{C}_1(\Gamma)$ and $\mathfrak{C}_1(\Gamma^\ast)$ be the set of homologically nontrivial cycles of $\Gamma$ and $\Gamma^\ast$ respectively. Then $\supp(\overline{X}_i)\in \mathfrak{C}_1(\Gamma^\ast)$ while $\supp(\overline{Z}_j)\in \mathfrak{C}_1(\Gamma)$. A CSS surface code state is stabilized by 
\ben
S= \spn{A_v, B_f, \overline{X}_1,\ldots, \overline{X}_l,\overline{Z}_{l+1},\ldots,\overline{Z}_k\bigg| \ba{l} v\in V(\Gamma)\\ 
f\in F(\Gamma)\ea},\label{eq:surfaceState}
\een
where we renumber the $\overline{X}_i$ and $\overline{Z}_j$ if necessary. 

Before, we leave this section, we need one more result that relates minors of graphs and their duals. There are two basic operations
we can perform on graphs:
\begin{compactenum}
\item Edge deletion
\item Edge contraction
\end{compactenum}
The graph obtained from $\Gamma$ by deleting an edge $e$ is denoted as $\Gamma\setminus e$ while the graph obtained
by edge contraction is denoted as $\Gamma/e$.  A graph obtained  from $\Gamma$ by a sequence of edge deletions and contractions is  called a (graph) minor of $\Gamma$. The following result on graph minors is well known. 
\begin{lemma}\label{lm:graphOps}
Supposing we have a graph $\Gamma$ and its dual $\Gamma^\ast$ and $e \in E(\Gamma)$. The following relations hold:
\begin{compactenum}[a.]
\item $(\Gamma\setminus e)^\ast = \Gamma^\ast/ e$
\item $(\Gamma/ e)^\ast = \Gamma^\ast \setminus e$
\end{compactenum}
\end{lemma}

\section{Counterexamples to the LU-LC Conjecture.} \label{sec:counterEx}

In this section, we address the first of the two questions we posed in the introduction, namely, are there finitely or infinitely
many counterexamples to the LU-LC conjecture. 
We show that there are infinitely many counterexamples to the $\lu$-$\lc$
conjecture and  that a counterexample exists for every odd length greater than 200.  We prove these
results constructively. 

Suppose that $Q$ is an $[[n,k,d]]$ quantum code. 
Let ${U}$ be an encoded gate for the code, therefore ${U}Q=Q$. If in addition, ${U}=\otimes_{i=1}^n{U}_i$
then we say that it is a transversal encoded gate\footnote{Transversal gates can also be defined for more than one block of a code,
but we will not need these general transversal gates in this paper.}. We need the following lemma; although straightforward, we include the proof for completeness.

\begin{lemma}\label{lm:encodedDLU} 
Let $\ket{\psi}$ and $\ket{\varphi}$ be two stabilizer states that are diagonal local unitary (DLU) equivalent, i.e.,
$
\ket{\varphi} =\bigotimes_{j=1}^n U_j \ket{\psi}, \mbox{ where } U_j = \left[\ba{cc}1&0\\0&e^{i\theta_j} \ea\right]
$.
Suppose that there exists an $[[m,1,d]]$ quantum code $Q$, which has a transversal implementation for some  $U_j$.
Encode the $j$th qubit of $\ket{\psi}$ and $\ket{\varphi}$ using $Q$ and denote the resulting states by 
$\ket{\overline{\psi}}$ and $\ket{\overline{\varphi}}$ respectively. Then $\lu(\overline{\psi}) = \lu(\overline{\varphi})$.
If  the transversal implementation of $U_j$ consists of diagonal unitaries, then $\ket{\overline{\psi}}$ and $\ket{\overline{\varphi}}$ 
are also DLU equivalent.
\end{lemma}
\begin{proof}
Let us rewrite the state $\ket{\psi}$ as $\ket{0}_j\ket{\psi^0}_{\sim j} + \ket{1}_j\ket{\psi^1}_{\sim j}$, where 
the subscripts $j$  and $\sim\!\!j$ denote the $j$th qubit and the remaining qubits respectively.
Encoding the $j$th qubit by $Q$  we obtain
\be
\ket{\overline{\psi}} = \ket{\overline{0}}_j   \ket{\psi^0}_{\sim j}+\ket{\overline{1}}_j   \ket{\psi^1}_{\sim j}.
\ee
Since $\ket{\varphi} =U \ket{\psi}$ we have
\be
\ket{\varphi} &=& U_j\ket{0}_j \bigotimes _{\stackrel{i=1}{i\neq j}}^n U_i\ket{\psi^0}_{\sim j}+ U_j\ket{1}_j \bigotimes _ {\stackrel{i=1}{i\neq j}}^n U_i\ket{\psi^1}_{\sim j}\\
&=&\ket{0}_j \bigotimes _{\stackrel{i=1}{i\neq j}}^n U_i\ket{\psi^0}_{\sim j}+ e^{i\theta_j}\ket{1}_j \bigotimes _ {\stackrel{i=1}{i\neq j}}^n U_i\ket{\psi^1}_{\sim j}.
\ee
Encoding the $j$th qubit we obtain
\be
\ket{\overline{\varphi}} &=&  \ket{\overline{0}}_j  \bigotimes _{\stackrel{i=1}{i\neq j}}^n U_i\ket{\psi^0}_{\sim j}+  e^{i\theta_j}\ket{\overline{1}}_j \bigotimes _{\stackrel{i=1}{i\neq j}}^n U_i\ket{\psi^0}_{\sim j}
\ee
Let the transversal implementation of  $U_j$ for $Q$ be $\overline{U}_j=\otimes_{i=1}^m \overline{U}_{j,i} $, then 
we can rewrite the above equation as 
\be
\ket{\overline{\varphi}} &=&\overline{U}_j\ket{\overline{0}}_j \bigotimes _{\stackrel{i=1}{i\neq j}}^n U_i\ket{\psi^0}_{\sim j} + 
\overline{U}_j\ket{\overline{1}}_j \bigotimes _{\stackrel{i=1}{i\neq j}}^n U_i\ket{\psi^0}_{\sim j}\\
&=&\overline{U}_j\bigotimes _{\stackrel{i=1}{i\neq j}}^n U_i (\ket{\overline{0}}_j \ket{\psi^0}_{\sim j} + 
 \ket{\overline{1}}_j  \ket{\psi^0}_{\sim j} )= \overline{U}  \ket{\overline{\psi}}.
\ee
Clearly, the states $\ket{\overline{\psi}}$ and $\ket{\overline{\varphi}}$ are LU equivalent. Further, they are DLU equivalent if $U_j$ has a diagonal transversal implementation, i.e., each of the $\overline{U}_{j,i}$ is diagonal, for $1\leq i\leq m$. 
\end{proof}

\begin{lemma}\label{lm:dluLCImpliesDLC}
Let two stabilizer states $\ket{\psi}$ and $\ket{\psi'} =\otimes_{i=1}^n U_i \ket{\psi}$ be diagonal local unitary equivalent, where each $U_i$ is a non-Clifford unitary. Then, all the minimal support elements of $S(\psi)$ and $S(\psi')$
must consist of $Z$-only operators.
If in addition, $\ket{\psi}$ and $\ket{\psi'}$ are local Clifford equivalent, 
then they must also be diagonal local Clifford equivalent. 
 \end{lemma}
\begin{proof}
By Lemma~\ref{lm:min2min} we know that $U$ maps the minimal support elements of $S(\psi)$ to minimal support elements of $S(\psi')$. 
The action of $U $ on $S(\psi)$ is to transform it as 
$US(\psi)U^\dagger$. Although $US(\psi)U^\dagger$ need not necessarily be in the Pauli group, the minimal support elements must be in the Pauli
group. Therefore, for any minimal support element $x \in S(\psi) $ we must have $U x U^\dagger \in  \mc{P}_n$. But
under the assumption that $U_i$ is a non-Clifford diagonal unitary, only  $Z$  of all the Pauli matrices is mapped back to a Pauli matrix. Assuming that $U_i=\text{diag}(1,e^{i\theta_i})$,  a Pauli matrix is transformed under 
conjugation as  $U \sigma U^\dagger$, where $\sigma\in \mc{P}_1$. If a minimal element in $S(\psi)$ contains either $X$ or $Y$
it cannot be mapped to a minimal element of $S(\psi')$ by an arbitrary non-Clifford diagonal unitary\footnote{The action of 
$\text{diag}(1,e^{i\theta})$ on $\mc{P}_1$ is to fix $I$ and $Z$ while  $X \mapsto \text{diag}(e^{-i\theta}, e^{i\theta})X$ and  
$Y \mapsto \text{diag}(e^{-i\theta}, e^{i\theta})Y$. Since the diagonal unitaries are non-Clifford, $\theta_i \not\in \spn{\pi/2}$, 
the only way we  can ensure that the minimal elements remain within the  Pauli group is to enforce that minimal elements consist of only $Z$-operators. }. Therefore all the minimal support elements
of $S(\psi)$ must be $Z$-only operators. Similarly the minimal support elements of $S(\psi')$ must be also $Z$-only operators. 

Without loss of generality we can assume that the stabilizer matrix of $\ket{\psi}$ is given by 
\ben
S(\psi) &=&\left[ \ba{cc|cc}I_k & P & Q & 0\\0 & 0 & P^t & I_{n-k} \ea\right], \label{eq:stabMat}
\een
where we have slightly abused the notation to denote both the stabilizer and the stabilizer matrix by $S(\psi)$.
Since the minimal elements are $Z$-only operators, any minimal element must be a minimal element\footnote{A minimal element   of $C$ is a codeword $c \in C $ such that there is no nonzero codeword $x$ in $C$ such that $\supp(x)\subsetneq \supp(c)$. } of the 
binary code generated by   $\left[ \ba{cc} P^t & I_{n-k} \ea\right]$. Denote this code as $C$. Since the codewords of $C$ correspond
to the $Z$ only elements in $S(\psi)$, we will sometimes refer to an element of $C$ as also being in $S(\psi)$. 
We claim that every minimal element of $C$ is also a minimal element of $S(\psi)$.  Let $a $ be  a 
minimal codeword of $C$. Suppose that $a$, or the element  in $S(\psi)$ corresponding to $a$, is not a minimal element of
$S(\psi)$. There exists some minimal element $b\in S(\psi)$ such that $\emptyset \neq \supp(b) \subsetneq \supp(a)$. 
We have already seen that every minimal element of $S(\psi)$
consists of $Z$-operators only. If so, then there exists a codeword $c$ in $ C$ such that $\supp(c) = \supp(b)$ and $\supp(c) \subsetneq \supp(a) $ contradicting the assumption that
$a$ was a minimal element of $C$. Therefore, every minimal element of $C$ is also a minimal element of $S(\psi)$.

We claim that the minimal elements of $S(\psi)$ span all the qubits i.e.,
\ben
\bigcup_{\stackrel{x \in S(\psi):}{x \text{ is minimal}}} \supp(x) = \{1,2,\ldots,n \}. \label{eq:minSuppUnion}
\een
To see this assume that there is some $j\in \{1,2,\ldots,n \}$ that is not in the support of any minimal element $x\in S(\psi)$.
As the minimal elements are $Z$-only operators, it must necessarily follow that 
$j$th column of $C$ is an all zero column. Without loss of generality, assume that this is the first column. Then it implies
that $s_1$, the first row of the stabilizer matrix in equation~\eqref{eq:stabMat} must be of the form $s_1=(1,0,\ldots,0|q_{11},q_{12},\ldots, q_{1k},0,\ldots,0)$ in order that it commute  with $S(\psi)$ \footnote{Recall that two elements of the Pauli group commute if and only if their
binary representations say $(a|b)$, $(c|d)$ are orthogonal with respect to the symplectic inner product, i.e.,
$a\cdot d+b\cdot c=0$.}.  Then $s_1$ has support only in the first $k$ qubits and not being
a $Z$-only operator must be non-minimal. Therefore the support of some minimal element must  be strictly contained in the
first $k$ qubits.  But note that $C$ is generated by  $\left[ \ba{cc} P^t & I_{n-k} \ea\right]$, thus every
minimal element has nonzero support in the last $n-k$ qubits. Therefore  no minimal support can be strictly in the support of
$s_1$, thereby making $s_1$ a minimal element and giving us a contradiction.  Hence equation~\eqref{eq:minSuppUnion} must hold.

Since $\ket{\psi}$ and $\ket{\psi'}$ are also local Clifford equivalent, there exists a local Clifford unitary ${K} = \otimes_{i=1}^n{K}_i$ such that $\ket{\psi'}={K}\ket{\psi}$; each of the $K_i\in \mc{K}_1$.  If $g$ is some minimal element in $S(\psi)$, then 
$g_i=Z$ for $i\in \supp(g)$. The local Clifford unitary also maps the minimal elements of $S(\psi)$ to 
the minimal elements of $S(\psi')$. Therefore $K_i g_i {K}_i^\dagger =Z $ for $i\in \supp(g)$. There are six 
possibilities \footnote{This is modulo sign. Recall that the action of Clifford group on a single qubit can be identified with the group of 
nonsingular $2\times 2 $ matrices over $\F_2$.} 
for the local Clifford unitaries---$\{ I , H, P, HP, PH, HPH \}$. Of these only if ${K}_i \in \{ I, P\}$ can the above requirement of 
mapping a minimal element to a minimal element is satisfied.  So if $i \in \supp(x) $ for some minimal element $x\in S(\psi)$, the associated Clifford
unitary is diagonal. Because equation~\eqref{eq:minSuppUnion} holds, we know every $i$  occurs in the support of some minimal element,
therefore $K_i$ is a diagonal local Clifford for all $i$.
Thus $\ket{\psi}$and $\ket{\psi'}$ are diagonal local Clifford equivalent. 
\end{proof}

A statement similar to second part of the previous lemma was attributed to Zeng in \cite{ji08}. With this preparation, we are now ready to give a 
constructive method to generate new counterexamples to the $\lu$-$\lc$ conjecture. 
\begin{theorem}\label{th:concat}
Let $\ket{\psi}$ be a CSS stabilizer state and $\ket{\psi'}$ another stabilizer state  
satisfying the following conditions:
\begin{compactenum}[\ref{th:concat}.a]
\item $\ket{\psi'} = \otimes_{i=1}^n U_i \ket{\psi}$ i.e. $\lu(\psi') = \lu(\psi)$ \label{th:concatConda}
\item $U_i$ is a diagonal non-Clifford unitary for all $i$ \label{th:concatCondb}
\item $\lc(\psi')\not= \lc(\psi)$ \label{th:concatCondc}
\item There exists an $[[m,1,d]]$ CSS code such that some $U_i$, say $U_n$, as an encoded gate,    has a diagonal non-Clifford transversal implementation i.e. $\overline{U}_n=\otimes_{j=1}^m F_j $ where $F_j$ is a non-Clifford diagonal unitary. \label{th:concatCondd}
\end{compactenum}
Then encoding the $n$th  qubit with the $[[m,1,d]]$ code results in encoded states $\ket{\overline{\psi'}} =\otimes_{j=1}^{n+m-1} U_i'\ket{\overline{\psi}}$ where each of the $U_i'=U_i$ for $1\leq i\leq n-1$  and $U_i'=F_{i-n+1}$ $n\leq i\leq n+m-1$ which 
satisfy  \ref{th:concat}.\ref{th:concatConda}--\ref{th:concatCondc}.
\end{theorem}
\begin{proof}
By \cite{dehaene03}, see also \cite[Theorem~5]{vanDenNest08},  we know that any stabilizer state can be represented in a ``standard from". In particular, a stabilizer state $\ket{\psi}$ in standard form can be associated to a subspace $S_\psi$ and a quadratic form $q_\psi(x)$ such that 
\be
\ket{\psi} = \sum_{x\in S_{\psi}} (-1)^{q_{\psi}(x)}\ket{x}.
\ee
Throughout this proof, we neglect the normalizing factors for convenience. We assume that both $\ket{\psi}$ and $\ket{\psi'}$ are in
standard form. In the present case, as the states $\ket{\psi}$ and $\ket{\psi'}$ are diagonal local unitary equivalent,
the same subspace $S$ is associated to both of them, and we can assume that $\ket{\psi}$ 
is associated to the trivial form $q_{\psi}(x) \equiv 0$ and $\ket{\psi'}$ is associated to the quadratic form $q(x)$.

We can write $q(x) = \tilde{q}(x)+q_n(x)$ as  the sum of two quadratic forms  $\tilde{q}(x)$ and $q_n(x)$ where
\be
\tilde{q}(x) = \sum_{1\leq i<j<n}q_{ij}x_ix_j \mbox{ and } q_n(x)=\sum_{j=1}^{n-1} q_{j n} x_jx_n.
\ee
We slightly abuse the notation and write $\tilde{q}(x)$,although $\tilde{q}( \cdot) $ does not explicitly depend on $x_n$.
Further we note that if $x_n=0$, then  $\tilde{q}(x) = q(x_1,\ldots,x_{n-1}, x_n=0)$, so we 
can rewrite the states $\ket{\psi}$ and $\ket{\psi'}$ using $\tilde{q}(x)$ and $q_n(x)$ as
\be
\ket{\psi} & =& \sum_{\stackrel{x\in S}{x_n=0}} \ket{x_1\cdots x_{n-1}} \ket{0}+\sum_{\stackrel{x\in S}{x_n=1}} \ket{x_1\cdots x_{n-1}}\ket{1},\\
\ket{\psi'}&=& \sum_{\stackrel{x\in S}{x_n=0}} (-1)^{\tilde{q}(x)}\ket{x_1\cdots x_{n-1}} \ket{0}+\\
& &\sum_{\stackrel{x\in S}{x_n=1}} (-1)^{\tilde{q}(x)+q_n(x)}  \ket{x_1\cdots x_{n-1}}\ket{1}.
\ee
Assume that we encode the $n$th qubit of $\ket{\psi}$ and $\ket{\psi'}$ using an $[[m,1,d]]_2$ CSS code $Q$ to obtain
the encoded states $\ket{\overline{\psi}}$ and $\ket{\overline{\psi'}}$.  We assume that $Q$ is derived from classical codes
$C$ and $D$ where $C\subset D \subseteq \F_2^m$.  The logical states of $Q$ are given by 
\be
\ket{\overline{0}} = \sum_{x\in C} \ket{x} \mbox{ and } \ket{\overline{1}} = \sum_{x\in D\setminus C} \ket{x}= \sum_{x  \in C} \ket{x+X_e},
\ee
where $X_e=  (a_1,\ldots,a_m )$ is any element in $D\setminus C$. It is in effect 
the encoded $X$ operator of $Q$, strictly speaking $(a_1,\ldots,a_m | 0,\ldots, 0 )$ 
is the binary representation of encoded $X$ operator. 
Let $Z_e=(b_1, \ldots, b_m)$ be any vector in $C^\perp\setminus D^\perp$.  Note that if $c\in C$, then  $c\cdot Z_e=0$
while  $c\cdot Z_e =1$ for any $c\in D\setminus C$, in particular this is true for $c=X_e$.

Then we can write the encoded states $\ket{\overline{\psi}}$ and $\ket{\overline{\psi'}}$ as
\be
\ket{\overline{\psi}}& =&\sum_{\stackrel{x\in S}{x_n=0}} \ket{x_1\cdots x_{n-1}} \ket{\overline{0}} +
 \sum_{\stackrel{x\in S}{x_n=1}} \ket{x_1\cdots x_{n-1}}\ket{\overline{1}}\\ 
&=&\sum_{\stackrel{x\in S}{x_n=0}} \ket{x_1\cdots x_{n-1}} \sum_{y\in C} \ket{y} +\\
&&\sum_{\stackrel{x\in S}{x_n=1}} \ket{x_1\cdots x_{n-1}}\sum_{y\in C} \ket{y+X_e}\\
\ket{\overline{\psi'}} &=& \sum_{\stackrel{x\in S}{x_n=0}} (-1)^{\tilde{q}(x)}\ket{x_1\cdots x_{n-1}}  \ket{\overline{0}}+\\
&&\sum_{\stackrel{x\in S}{x_n=1}} (-1)^{\tilde{q}(x) +q_n(x)}  \ket{x_1\cdots x_{n-1}}\ket{\overline{1}}\\
&=& \sum_{\stackrel{x\in S}{x_n=0}} (-1)^{\tilde{q}(x)}\ket{x_1\cdots x_{n-1}}  \sum_{y\in C} \ket{y}  +\\
&&\sum_{\stackrel{x\in S}{x_n=1}} (-1)^{\tilde{q}(x) +q_n(x)}  \ket{x_1\cdots x_{n-1}}\sum_{y\in C} \ket{y+X_e}
\ee
Let $E =(e_1,\ldots, e_n) \in S$ be a vector such that $e_n=1$.  Denote by $\overline{E}=(e_1,\ldots,e_{n-1})$, the vector obtained by dropping the $n$th coordinate. 
As the encoded states are also stabilizer states we can associate a subspace $\overline{S}$ to them as
\be
\overline{S} &= & \spn{S_0\oplus C, (\overline{E}|X_e) }
\ee
where $(\overline{E} | X_e)= (e_1,\ldots, e_{n-1}, a_1,\ldots, a_m)$ and 
\be
S_0&=& \{ (x_1,\ldots, x_{n-1}) \mid x \in S \mbox { and } x_n=0\}
\ee
Additionally, we can associate a quadratic form $\overline{q}( \cdot )$ to $\ket{\overline{\psi'}}$ as
\be
\overline{q}(x_1,\ldots, x_{n-1}, y_1,\ldots,  y_m)& =& 
\sum_{1\leq i<j<n} q_{ij} x_ix_j +\\
 \sum_{j=1}^{n-1} q_{jn} x_j (b_1 y_1 + \cdots+b _m y_{m}), 
\ee
where $(b_1,\ldots, b_m)\in C^\perp\setminus D^\perp$.  Although, $\overline{q}$ does not explicitly  depend on $x_n$ we   write $\overline{q}(x,y)$ for $\overline{q}(x_1,\ldots, x_{n-1}, y_1,\ldots, y_m)$. 
If $x\in S $ and $x_n=0$, then $y$ must be in $C$. Then $b_1y_1+\cdots + b_my_m=Z_e \cdot y = 0$ as $Z_e \in C^\perp$
and $\overline{q}(x,y)$ reduces to $\tilde{q}(x) =q(x)|_{x_n=0}$. If $x\in S$ and $x_n=1$, then we must have $y \in D\setminus C$ giving
$b_1y_1+\cdots + b_my_m=Z_e \cdot y = 1$. In this case $\overline{q}(x,y)$ reduces to $\tilde{q}(x) + q_n(x)=q(x)|_{x_n=1}$.
Thus
\be
\ket{\overline{\psi'}} = \sum_{x\in S_0\oplus C} (-1)^{\overline{q}(x)}\ket{x} +\sum_{x\in (S_0\oplus C)+(\overline{E}|X_e)} (-1)^{\overline{q}(x)}\ket{x}
\ee
By assumption  the code $C$ has a diagonal non-Clifford transversal  implementation of $\overline{U}_i$; without loss of generality we can assume that $i=n$. 
Then by Lemma~\ref{lm:encodedDLU}, the states $\ket{\overline{\psi}}$ and $\ket{\overline{\psi'}}$
are diagonal $\lu$ equivalent. This proves that $\ket{\overline{\psi}}$ and $\ket{\overline{\psi'}}$ satisfy
 \ref{th:concat}.\ref{th:concatConda} and \ref{th:concat}.\ref{th:concatCondb}.

We claim that  $\ket{\overline{\psi}}$ and $\ket{\overline{\psi'}}$ are not LC equivalent. 
Suppose that they are LC equivalent. Then by Lemma~\ref{lm:dluLCImpliesDLC},  $\ket{\overline{\psi}}$ and $\ket{\overline{\psi'}}$
are DLC equivalent. In such a case there exist complex numbers $c_j \in \spn{i} $, $1\leq j\leq n+m-1$ such that 
\be
\prod_{j=1}^{n+m-1}c_j^{s_j} = (-1)^{\overline{q}(s)} \mbox{ for all } s\in \overline{S}.
\ee
Alternatively, letting $s=(x|y)$, we must have 
\ben
\prod_{j=1}^{n-1}c_j^{x_j}\prod_{j=1}^{m}c_{j+n-1}^{y_j} &= &(-1)^{\overline{q}(x,y)} \label{eq:qfpEqns}\\
&=&(-1)^{\tilde{q}(x) +\sum_{j=1}^{n-1} q_{jn}x_j(b_1y_1+\cdots+b_my_{m})} \nonumber,
\een
for all $ (x|y) \in \overline{S}$.
Consider an $s\in S_0\oplus C \subset \overline{S}$, in particular let $s=(x|y)$  where $y=0$ 
and $x$ is any element such that $ (x_1,\ldots, x_{n-1},0) \in S$. Then equation \eqref{eq:qfpEqns} reduces to
\ben
\prod_{j=1}^{n-1}c_j^{x_j}& =& (-1)^{\overline{q}(x,0)} =(-1)^{\tilde{q}(x)} \\
& =&(-1)^{q(x)} \mbox{ for all }  (x_1,\ldots, x_{n-1},0) \in S.\label{eq:qfp1}
\een
If $s=(x|y)$ is in  $\overline{S}\setminus (S_0\oplus C)$, then $(x_1,\ldots,x_{n-1},1) \in S$ and 
$y\in D\setminus C$. Choosing $y=\overline{X}$, we get 
\be
\prod_{j=1}^{n-1}c_j^{x_j} \prod_{j=1}^{m}c_{j+n-1}^{y_j}&=& 
(-1)^{\tilde{q}(x) +\sum_{j=1}^{n-1} q_{jn}x_j(b_1y_1+\cdots+b_my_{m})} \\
&=&
(-1)^{\tilde{q}(x) +\sum_{j=1}^{n-1} q_{jn}x_j (Z_e \cdot X_e)} \\
&=&(-1)^{\tilde{q}(x) +\sum_{j=1}^{n-1} q_{jn}x_j } \\
&=&(-1)^{\tilde{q}(x) +\sum_{j=1}^{n-1} q_{jn}x_jx_n } 
\ee
Letting $C_j=c_j$ for $1\leq j<n$ and $C_n=\prod_{j=n}^{n+m-1}c_j^{y_j} $ we obtain 
\ben
\prod_{i=1}^{n}C_i^{x_i}  &=& (-1)^{q(x)} \mbox{ for all } (x_1, \ldots, x_{n-1}, 1)\in S.\label{eq:qfp2}
\een
Since $c_j\in \spn{i}$, $C_j$ are also in $\spn{i}$. Together equations~\eqref{eq:qfp1} and \eqref{eq:qfp2} 
imply that for 
all $x\in S$, there exist $C_j\in \spn{i}$ such that 
\be
\prod_{j=1}^n C_j^{x_j} =(-1)^{q(x)} \mbox{ for all } x\in S.
\ee
It follows that the states $\ket{\psi}$ and $\ket{\psi'}$ are diagonal local Clifford equivalent, (see \cite{ji08}). 
But this contradicts that  $\ket{\psi}$ and $\ket{\psi'}$ are not local Clifford equivalent. Therefore it must be that $\ket{\overline{\psi}}$
and $\ket{\overline{\psi'}}$ are not LC equivalent. Thus they furnish another counterexample for the LU-LC conjecture. 
\end{proof}
We point out that the assumption that $\ket{\psi}$ is a CSS stabilizer state in Theorem~\ref{th:concat} is not really a limitation, because of the following reason. Supposing that $\ket{\psi}$ and $\ket{\psi'}$ are both not CSS stabilizer states. Let us
associate them to two distinct quadratic forms $q_{\psi}(x)$
and $q_{\psi'}(x)$, as the states are not local Clifford equivalent.
\be
\ket{\psi} = \sum_{x\in S} (-1)^{q_{\psi}(x)}\ket{x} \mbox{ and } \ket{\psi'} =\sum_{x\in S} (-1)^{q_{\psi'}(x)} \ket{x}.
\ee
We can consider the following states $\ket{\varphi}$ and $\ket{\varphi'}$ instead of $\ket{\psi}$ and
$\ket{\psi'}$.
\be
\ket{\varphi} = \sum_{x\in S} \ket{x} \mbox{ and } \ket{\varphi'} =\sum_{x\in S}\ket{x } (-1)^{q(x)} \ket{x},
\ee
where $ q(x) =q_{\psi}(x)+q_{\psi'}(x)= \sum_{1\leq i<j  \leq n } q_{ij} x_i x_j$.
The states $\ket{\varphi}$  and $\ket{\varphi'}$ are related by the same local diagonal unitaries as $\ket{\psi}$ and
$\ket{\psi'}$,  i.e., $\ket{\varphi'} = \otimes_{i=1}^n U_i \ket{\varphi}$. Further, they are also a pair of states that violate the
$\lu$-$\lc$ conjecture. If they are not,  then $\ket{\varphi'} $ and $ \ket{\varphi}$ are $\lc$ equivalent. But by 
Lemma~\ref{lm:dluLCImpliesDLC} they are also DLC equivalent which in 
turn implies the DLC equivalence of $\ket{\psi}$ and  $\ket{\psi'}$ contradicting that  
$\ket{\psi}$ and  $\ket{\psi'}$ are not LC equivalent. We can see that $\ket{\varphi}$ and $\ket{\varphi'}$ satisfy the conditions \ref{th:concat}.\ref{th:concatConda}--\ref{th:concatCondd}.
Consequently, we can use these states to generate new counterexamples
in Theorem~\ref{th:concat}.

Additionally, we could also lift the assumption on the use of a CSS code in Theorem~\ref{th:concat}, provided the non-CSS 
$[[m,1,d]]$ code has 
a non-Clifford transversal gate that can be implemented using diagonal non-Clifford gates. 
An alternate proof of Theorem~\ref{th:concat} seems to be possible based on \cite{vanDenNest04}. 
\begin{corollary}\label{co:concatRM}
There are infinitely many stabilizer states for which the LU-LC conjecture does not hold. In particular, there exist stabilizer states
which are LU equivalent but not LC equivalent under the following conditions:
\begin{compactenum}[a.]
\item distance two and $n\geq 28$
\item distance $\geq 3$ and odd length $n\geq 195$ 
\end{compactenum}
\end{corollary}
\begin{proof}
We note that the 27 qubit counterexample provided in \cite{ji08} provides  a starting point for the application of
Theorem~\ref{th:concat}. In this case the two locally equivalent stabilizer states are related by diagonal non-Clifford unitaries and
 are odd powers of the  $T$ gate where $T =\text{diag} (1, e^{i\pi/4})$. We know that the
$T$ gate has a  diagonal non-Clifford transversal implementation for the $[[15,1,3]]$ Reed-Muller code, while $\sqrt{T} = \text{diag}(1,e^{i\pi/8})$ has a diagonal non-Clifford implementation for the $[[31,1,3]]$ Reed-Muller code and the 
$[[2,1,1]]$ code. 
 So all these codes permit a transversal
implementation of the $T$ gate.  In fact their transversal implementations use $T$ gate on each qubit. So we can concatenate with 
either code to obtain a new counterexample of length $41$ and $57$ and $28$ respectively. 
The resulting state obtained by encoding with any of these code also satisfies \ref{th:concat}.\ref{th:concatCondd} in addition to the conditons  \ref{th:concat}.\ref{th:concatConda}--\ref{th:concatCondc}.  If we repeatedly use the $[[2,1,1]]$ code, then we obtain
distance two counterexamples for any $n\geq 28$.  

If we use a combination of $[[15,1,3]]$ and $[[31,1,3]]$ and  repeatedly apply Theorem~\ref{th:concat} 
to these states to obtain a series of stabilizer states of length $27+14i+30j$ which are $\lu$-equivalent but not $\lc$-equivalent. 
It remains to show that all the odd lengths beyond $n\geq 195 $ can be attained by this concatenation. 
Now assume that we have concatenated $i$ times with $[[15,1,3]]$ and $j$ times with $[[31,1,3]]$ to obtain $n(i,j)=27+14i+30j$
length counterexamples. We can ignore the offset due to 27 and consider the linear combinations of $14i+30j=2c$. A solution for
this equation is $i_0=-2c$ and $j_0=c$. We know from number theory, see \cite[Theorem~5.1]{niven01}, 
 that there exist infinitely many solutions for this equation 
given by 
\be
i = i_0 + 15 k \mbox{ and } j = j_0 - 7k, 
\ee
where $k$  is an integer. 
Since we are interested in only those combinations for which $i, j \geq 0$, we require that
\be
15k \geq -i_0 \mbox{ and } 7k \leq j_0. 
\ee
This implies that 
\be
2c/15 \leq  k \leq c/7
\ee
Assuming that $c=7q+r$, where $0\leq r\leq 6$, we find that 
\be
(14q+2r)/15 \leq k \leq q+r/7.
\ee
An integral solution of $k=q$ is possible if $(14q+2r)/15 \leq q$ i.e. $q\geq 2r$ or $q\geq 12$. 
Thus we obtain $i=i_0+15q=-2c+15q$ and $ j=j_0-7q=c-7q=r$. We obtain $2c = 14q+2r \geq 14\cdot 12 =168$. So for all
odd lengths $n\geq 27+168=195$, there exists a counterexample to the $\lu$-$\lc$ conjecture. 
\end{proof}

\section{LU-LC Equivalence of Cluster States} \label{sec:cluster}
Stabilizer states can be derived from graphs in two important ways. In one method we associate the edges to qubits. Such
stabilizer states will be studied in the next section. In this section we study stabilizer states wherein the qubits are identified
with the vertices of a graph. Such states are also termed graph states. A special class of these
graph states are the cluster states. We turn our attention to cluster states and consider their equivalence classes. Our interest in 
cluster states stems from the fact that the cluster states form a universal resource for measurement based quantum computation.

\begin{lemma}
Let $\Gamma$ be an $m\times n $ rectangular grid, where $m,n \geq 5$. Then the local unitary and local Clifford
equivalence classes of the graph state $\ket{G}$ are the same.
\end{lemma}
\begin{proof}
The stabilizer of the graph state is  given by
\be
S = \langle K_v \mid v\in V(\Gamma) \rangle \mbox{ where } K_v = X_v \prod_{y\in N(x) } Z_{y},
\ee
and  $N(x)$ denotes the neighbors of $x$.  
If a generator $K_v$ is  not a minimal support element of $S$, then there exists
a linear combination of the generators $\{ K_v\}\cup \{ K_x \mid x\in N(v) \}$ such that  its support is strictly contained in $\supp(K_v)$.
We partition the vertices of the graph into three different sets:
\begin{compactenum}
\item $V_i$:  vertices that are not connected to the boundary at all
\item $V_o$: vertices on the boundary of the grid
\item$V_m$: vertices which are connected to the boundary by at least one edge
\end{compactenum}
Consider a generator $K_v$, where $v\in V_i$.  
\begin{center}
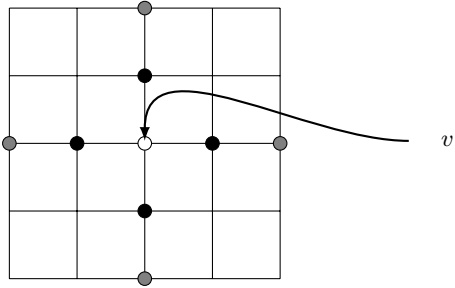
\begin{figure}[h]
\begin{tikzpicture}[scale=.9,every node/.style={minimum size=1cm},on grid]		
    \begin{scope}[
    	yshift=90,every node/.append style={
    	yslant=0.00,xslant=0},yslant=0,xslant=0
    	             ]
   	\draw[step=10mm, black] (0,0) grid (4,4);
     \draw [fill=white](2,2) circle (.1) ;
   	 \draw [fill=black](2,1) circle (.1) ;
        \draw [fill=black](2,3) circle (.1);
        \draw [fill=black](1,2) circle (.1);
        \draw [fill=black](3,2) circle (.1);
 	 \draw [fill=gray](2,0) circle (.1) ;
        \draw [fill=gray](2,4) circle (.1);
        \draw [fill=gray](0,2) circle (.1);
        \draw [fill=gray](4,2) circle (.1);
  \end{scope}
    	
     \draw[-latex,thick](5.9,5.2)node[right]{$v$}
        to[out=180,in=90] (2,5.2);
\end{tikzpicture}
 \caption{$v$ inside the boundary and without neighbor(s) on the boundary}
\end{figure}
\end{center}
Let $N(v)=\{ a, b, c, d \}$, then we can see that every neighbor of $v$ is connected to a vertex which is not 
connected to the remaining three neighbors of $v$. So any linear combination involving $K_v, K_a, K_b, K_c, K_d$
will have a support outside $\supp(K_v)$ unless it is equal to $K_v$.

Next consider a vertex on the boundary of $\Gamma$. Since $m,n \geq 5$, we need to consider three cases shown below 
when $v$ is on the corner
of the grid, when one of $v$'s neighbors is a corner and when no neighbor is  a corner. 
\begin{center}
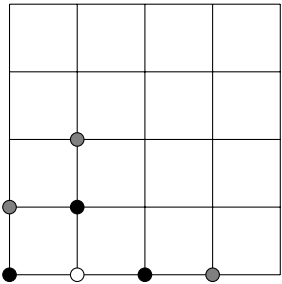
\begin{figure}[h]
\begin{tikzpicture}[scale=.9,every node/.style={minimum size=1cm},on grid]		
    \begin{scope}[
    	xshift=150, yshift=90,every node/.append style={
    	yslant=0.00,xslant=0},yslant=0,xslant=0
    	             ]
    	\draw[step=10mm, black] (0,0) grid (4,4);
     \draw [fill=white](1,0) circle (.1) ;
   	 \draw [fill=black](0,0) circle (.1) ;
        \draw [fill=black](2,0) circle (.1);
     		\draw [fill=black](1,1) circle (.1);
  	    \draw [fill=gray](1,2) circle (.1);
     \draw [fill=gray](3,0) circle (.1);
        \draw [fill=gray](0,1) circle (.1);
  \end{scope}	
  \end{tikzpicture}
  \caption{$v$ on the boundary with  neighbor(s) on a corner}
   \end{figure}
   \begin{figure}[h]
  \begin{tikzpicture}[scale=.9,every node/.style={minimum size=1cm},on grid]	
      \begin{scope}[
    	yshift=90,every node/.append style={
    	yslant=0.00,xslant=0},yslant=0,xslant=0
    	             ]
    	\draw[step=10mm, gray] (0,0) grid (4,4);
     \draw [fill=white](0,0) circle (.1) ;
   	 \draw [fill=black](1,0) circle (.1) ;
        \draw [fill=black](0,1) circle (.1);
  	 \draw [fill=gray](2,0) circle (.1) ;
          \draw [fill=gray](0,2) circle (.1);
  \end{scope}	
    \begin{scope}[
    	xshift=150, yshift=90,every node/.append style={
    	yslant=0.00,xslant=0},yslant=0,xslant=0
    	             ]
    	\draw[step=10mm, black] (0,0) grid (4,4);
     \draw [fill=gray](2,2) circle (.1) ;
   	 \draw [fill=gray](0,0) circle (.1) ;
        \draw [fill=black](1,0) circle (.1);
        \draw [fill=white](2,0) circle (.1);
        \draw [fill=black](3,0) circle (.1);
     \draw [fill=black](2,1) circle (.1);
   	 \draw [fill=gray](4,0) circle (.1) ;
  \end{scope}	
  \end{tikzpicture}
  \caption{$v$ on the boundary and without neighbors on a corner}
\end{figure}
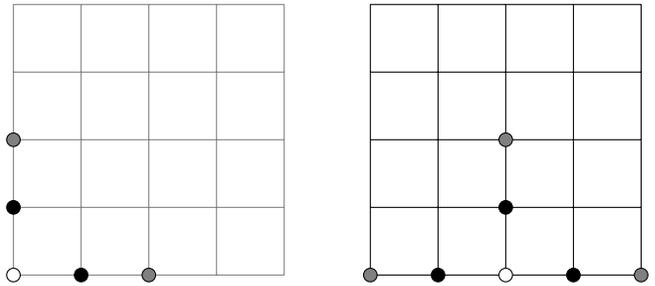
\end{center}

Finally, when $v$ is in $V_m$, then we need to consider the following cases depending on whether $v$ has one or two 
neighbors on the boundary of the grid.
\begin{center}
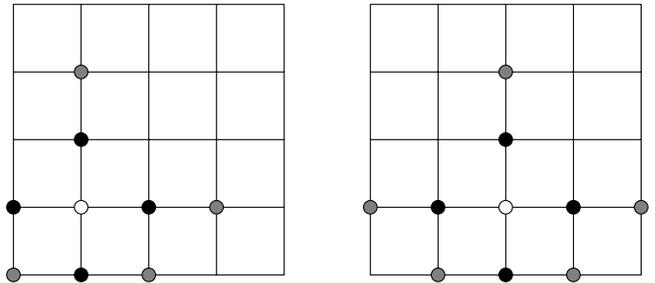
\begin{figure}[h]
\begin{tikzpicture}[scale=.9,every node/.style={minimum size=1cm},on grid]		
    \begin{scope}[
    	yshift=90,every node/.append style={
    	yslant=0.00,xslant=0},yslant=0,xslant=0
    	             ]
    	\draw[step=10mm, black] (0,0) grid (4,4);
     \draw [fill=gray](0,0) circle (.1) ;
     \draw [fill=white](1,1) circle (.1) ;
   	 \draw [fill=black](1,0) circle (.1) ;
        \draw [fill=black](0,1) circle (.1);
  	 \draw [fill=gray](2,0) circle (.1) ;
         \draw [fill=black](2,1) circle (.1);
  	 \draw [fill=gray](3,1) circle (.1) ;
          \draw [fill=black](1,2) circle (.1);
          \draw [fill=gray](1,3) circle (.1);
  \end{scope}
    \begin{scope}[
    	xshift=150, yshift=90,every node/.append style={
    	yslant=0.00,xslant=0},yslant=0,xslant=0
    	             ]
    	\draw[step=10mm, black] (0,0) grid (4,4);
     \draw [fill=gray](2,3) circle (.1) ;
   	 \draw [fill=black](2,2) circle (.1) ;
        \draw [fill=black](2,0) circle (.1);
        \draw [fill=black](1,1) circle (.1);
  	 \draw [fill=black](3,1) circle (.1) ;
        \draw [fill=white](2,1) circle (.1);
     \draw [fill=gray](4,1) circle (.1);
        \draw [fill=gray](0,1) circle (.1);
        \draw [fill=gray](1,0) circle (.1);
        \draw [fill=gray](3,0) circle (.1);
  \end{scope}	
  \end{tikzpicture}
  \caption{$v$ inside the boundary but with at least one neighbor on the boundary}
  \end{figure}
\end{center}
It can be easily verified that in each of the above cases, whenever there is a linear combination of the generators
in the support of $v$, the resulting linear combination always has support outside $\supp(v)$ unless the combination is
trivially equal  to $K_v$. Thus every generator $K_v$ is a minimal element in $S$. Thus $M(\ket{G})= S$ itself and by Lemma~\ref{lm:msc}, we have $\lu(\ket{G}) =\lc(\ket{G})$. 
\end{proof}

Zeng et al., \cite{zeng07} showed that if a graph $\Gamma$ has no 3 or 4 cycles, then $\lu(\ket{G}) = \lc(\ket{G})$. From this
it follows that the graph states associated to honeycomb (i.e., hexagonal) lattices also satisfy the LU-LC conjecture. 
These results for the cluster states and the hexagaonal lattices are somewhat surprising in that these states  are universal resources
for measurement based quantum computation.

\section{LU-LC Equivalence of Surface Code States}\label{sec:surfaceStates}

In this section we show that surface code states cannot be counterexamples to the $\lu$-$\lc$ conjecture. 
These results prepare the way to make the connection to matroids.
We consider surface codes derived from connected graphs that have no loops or coloops. 
Additionally, we require the graphs to have no 2-cycles or 2-cocycles and are thus free from 
Bell pairs. 
The following lemma is known in the context of surface codes, but we record it for our later use.
\begin{lemma}\label{lm:cyclesOnly}
Let $\ket{\psi}$ be a surface code state associated to a graph $\Gamma$  without loops or coloops. 
Let $g$ be an element in $C(S(\psi))$, the centralizer of  $S({\psi})$ such that $g$ consists of only $Z$  (or  $X$) operators alone. Then $\supp(g)$ is an union of cycles   (or cocycles) in $\Gamma$.
In particular, the face (vertex) operators of $\Gamma$ generate only union of cycles (cocycles) of $\Gamma$.
\end{lemma}
\begin{proof}
Assume that  $g$ is a $Z$-only operator in $C(S(\psi))$ such that its support is not a cycle in $\Gamma$. Since $g$ consists of only $Z$ operators it must be generated by the face operators of $\Gamma$ and the encoded $Z$ operators of the surface  code associated to $\Gamma$. If the support of $g$
is not an union of cycles then there exists a vertex $v$ such that an odd number  of the edges of $\supp(g)$  are incident on $v$. Then the vertex operator $A_v$ (consisting of $X$ operators alone) does not commute with $g$ as it overlaps with $g$ over odd number of qubits. Therefore $g$ cannot be an element of 
$C(S(\psi))$, as all elements in $C(S(\psi))$ commute with vertex operators.  By a similar argument but considering the 
 dual graph we can also show that all the $X$-only operators in $C(S(\psi))$ must be cycles in $\Gamma^\ast$.
\end{proof}
Requiring that $\Gamma$ has no coloops is equivalent to stating that it has no cut set of size one. 
Lemma~\ref{lm:cyclesOnly} can be extended to include graphs with loops and coloops, however this is sufficient for us. 
In the following we often denote the stabilizer of a quantum code $Q$ by $S(Q)$.
Given a subset $\omega \subseteq \{1,2,\ldots,n\}$,  and an abelian subgroup $S \leq C(S(Q))$ we denote by $A_\omega^S =|\{ g\in S: \supp(g)=\omega\}|$.

\begin{lemma}\label{lm:minimalType}
Suppose  $Q$ is a surface code associated to a graph $\Gamma$ without cycles or cocycles of length $\leq 2$. Let $S$ be an abelian subgroup of $C(S(Q))$ that contains $S(Q)$. Then for any  $A_v\in S$ and $B_f\in S$ we there exist minimal elements
$\{g_1^v,\ldots, g_{k_v}^v \} \subset S $ and $\{h_1^f,\ldots, h_{k_f}^f \}\subset S$ such that 
\begin{compactenum}[i)]
\item Each  $g_i^v$ is a $X$-only operator and $\supp(A_v) =\cup_{i=1}^{k_v} \supp(g_i^v)$.
\item $A_{\supp(g_i^v)}^S=1$.
\item  Each $h_i^f$ is a $Z$-only operator and $\supp(B_f) = \cup_{i=1}^{k_f} \supp(h_i^f)$.
\item $A_{\supp(h_i^f)}^S=1$.
\end{compactenum}
Further, for any $j\in \{1,2,\ldots, n \}$, there exists some $A_v$ and $B_f$ such that $j\in \supp(A_v)\cap \supp(B_f)$, 
where $n=|E(\Gamma)|$, the total number of qubits.
\end{lemma}
\begin{proof}
If $A_v$ is minimal then i) holds trivially for $A_v$ by letting $g_1^v=A_v$. 
Suppose that $A_v$ is a non-minimal element in $S$. There exists a minimal element $g_1^v\neq I$ in $S$ such that 
$\supp(g_1^v) \subsetneq \supp(A_v)$. Now as $Q$ is a CSS code, $g_1^v$  can be written as 
$g_1^v=g_xg_z$,  for some $g_x,g_z\in C(S(Q))$ where both $g_x$ and $g_z$ are not simultaneously trivial. 
It follows that the supports of both $g_x$ and $g_z $ must be strictly contained in the support of $A_v$. In particular, 
 $\supp(g_z) \subsetneq \supp(A_v)$. By Lemma~\ref{lm:cyclesOnly}, $g_z$ must be a cycle of $\Gamma$. However, all the edges of $g_z$ being a subset of $\supp(A_v)$ must be incident on the same vertex. Then $\supp(g_z)$ can be
cycle  only if  it contains 2-cycles. But this contradicts that $\Gamma$ does not have 2-cycles. Therefore, $g_z=I$ and $g_x= g_1^v\neq I$. Now consider the element $g'=A_vg_1^v$, it is an $X$-only operator whose support is given by $\supp(A_v)\setminus \supp(g_1^v)$. Now $g'$ is either minimal or not. If it is minimal then by relabeling $g'=g_2^v$ we can write $\supp(A_v)= \supp(g_1^v)\cup \supp(g_2^v )$ and we are done. On the other hand if $g'$ is not minimal, we can repeat the same process as with $A_v$ and we will eventually
end up with a set of elements $\{ g_1^v,\ldots, g_{k_v}^v \}$ such that 
$\supp(A_v)= \cup_{i=1}^{k_v} \supp(g_i^v)$. This shows that i) holds.

Let $\omega=\supp(g_i^v)$. Suppose that $A_{\omega}^S\neq 1$. Then there exists a $g_i'\in S$ such that $g_i'\neq g_i$ and
$\supp(g_i')=\omega$. Since $\omega \subsetneq \supp(A_v)$, it follows that $g_i'$ is also a minimal element within the support
of $A_v$. But we have already seen that every minimal element in the support of $A_v$ must consist of $X$-only operators.
Therefore $g_i'=g_i^v$ contradicting that $g_i'\neq g_i^v$. Thus $A_{\supp(g_i^v)}^S=1$,  proving ii).

We claim that $j\in \{1,2,\ldots,n \}$ occurs in the support of $A_v$ for some $v\in V(\Gamma)$.
Since an edge is incident on at most two vertices, if $j\not\in\supp(A_v)$ for any $v\in V(\Gamma)$, 
this means that both ends of the edge associated to $j$ must be incident on the same vertex. This implies that $\Gamma$ has loops
contrary to our assumptions. So every $j$ occurs in the support of some site operator $A_v$.

By a similar argument but working in the dual graph $\Gamma^\ast$, we can show that $\supp(B_f) =\cup_{i=1}^{k_f}  \supp(h_i)$
for some $Z$-only minimal elements in $S$ and that  $j$ occurs in the support of some face operator $B_f$ as long as there are no coloops. 
\end{proof}

\begin{theorem}\label{th:surfaceCode}
Let $Q$ be the surface code associated to a graph $\Gamma$ cycles or cocycles of length $\leq 2$.
Let $S(Q)$ be the stabilizer of $Q$
and $C(S(Q))$, the centralizer of $S$. Then for any surface code (stabilizer) state $\ket{\psi}$ we have $\lu(\psi)=\lc(\psi)$. 
\end{theorem}
\begin{proof} The stabilizer of a general surface code state  is a subgroup of $C(S(Q))$  and contains $S(Q)$. 
\be
S(\psi) =\spn{S, L_1, L_2, \ldots, L_k \bigg|\ba{l} 
L_i \in C(S(Q))\setminus SZ(S(Q)) \\ \mbox { and } L_iL_j =L_j L_i\ea},
\ee
where $k$ is such that $S(\psi)$ is the stabilizer of $\ket{\psi}$. Let $n=|E(\Gamma)|$, the total number of qubits. 
From Lemma~\ref{lm:minimalType}, we know that every vertex operator and face operator are such that 
there exist $X$-only  minimal elements $\{ g_1,\ldots, g_{k_v} \}\in S(\psi)$ and $Z$-only minimal elements 
$\{ h_1,\ldots, h_{k_f} \} \in S(\psi)$ such that $\cup_{i=1}^{k_v}\supp(g_i) =\supp(A_v)$ and $\supp(B_f)=\cup_{i=1}^{k_f}\supp(h_i)$.
Since for every $j$ we have $j\in \supp(A_v) \cap \supp(B_f)$ for some $v\in V(\Gamma)$ and $f\in F(\Gamma)$,
there exist minimal elements $g,h$ such that $g_j=X$ and $h_j=Z$. This means that on every qubit, all three $X, Y, Z $ 
occur on every qubit  in $M(\psi)$.  Since there are no 2-cycles or 2-cocycles, $S(\psi)$ cannot contain elements of the 
form $X_iX_j$ or $Z_iZ_j$. Since $Q$ is a CSS code  if $X_iZ_j$ or $X_iY_j$ are in $C(S(Q))$, then it follows that $\Gamma$
has loops or coloops, contradicting the assumptions on $\Gamma$. Thus $\ket{\psi}$ is free of Bell pairs. 
By Lemma~\ref{lm:msc},  it follows that $\lu(\psi)= \lc(\psi)$. 
\end{proof}

\subsection{Application: On the Existence of Non-Clifford Transversal Gates}
From the point of fault tolerant quantum computing (encoded) transversal gates are very desirable for quantum codes. Transversal gates
within the Clifford group can be found using the symmetries of the stabilizer. Transversal gates outside the Clifford group 
are known only for a handful of quantum codes, the prominent being the quantum Reed-Muller codes \cite{zeng07}. The primary motivation for non-Clifford transversal gates was to find a universal set of encoded transversal gates. That such a universal set does not exist for stabilizer codes has been shown in \cite{zeng07a} and  a similar result for general quantum codes was proved in 
\cite{eastin09}.

Finding transversal gates outside the Clifford group has turned to be very
difficult. In fact, Rains \cite{rains99} has shown that non-Clifford transversal gates do not exist
for $GF(4)$-linear quantum codes. Surface codes are purely additive quantum codes, and in this
section we show that under some restrictions, 
surface codes preclude the existence of non-Clifford transversal gates.

\begin{lemma}\label{lm:cond1OnU}\cite[Lemma~1]{zeng07a}
Let $Q$ be a stabilizer code. If $U=U_1\otimes \cdots \otimes U_n$ is a logical gate for $Q$,
then $U_{\omega} \rho_{\omega}(Q)  = \rho_{\omega}(Q) U_{\omega}$ for all $\omega \subseteq \{1,2,\ldots, n\}$, where $U_{\omega}  = \otimes_{i\in \omega} U_i$,
\ben
\rho_{\omega} =\frac{1}{B_{\omega}(Q)}\sum_{\stackrel{s\in S(Q)}{\supp(s)\subseteq \omega}} \Tr_{\bar{\omega}}(s); \\
\quad B_\omega(Q) =|\{s\in S(Q) :   \supp(s) \subseteq \omega \}|
\label{eq:cond1OnU}
\een
 for all $\omega$.
More generally, if $Q'$
is another  stabilizer code such that $U$ is a local equivalence from $Q$ to $Q'$, then 
$U_{\omega} \rho_{\omega}(Q)  = \rho_{\omega}(Q') U_{\omega}$,
\end{lemma}

\begin{theorem}
Let $Q$ be a  surface code  such that the associated graph does not have 
any cycles or cocycles of length $\leq 2$. 
Then $Q$ does not have any transversal encoded non-Clifford gate, $U=U_1\otimes U_2\otimes \cdots \otimes U_n$.
\end{theorem}
\begin{proof}
We assume that $Q$ is derived from a graph $\Gamma$. 
Let $g$ be a site operator or face operator.
By Lemma~\ref{lm:minimalType}, we know that $\supp(g)=\cup_{i\in k_g} \supp(m_i^g)$, for some minimal elements 
$\{m_1^g,\ldots, m_{k_g}^g \} \subset S(Q)$. 
Let $\omega= \supp(m_i^g)$, then we have that $A_{\omega}^S(Q)=1$. 
Now computing $\rho_{\omega}$ as given by \eqref{eq:cond1OnU} we have
\be
\rho_{\omega} =\frac{1}{2}\sum_{\stackrel{s\in S(Q)}{\supp(s)\subseteq \omega}} s= \frac{1}{2} \left(\Tr_{\bar{\omega}}(I)+\Tr_{\bar{\omega}}(m_i^g)\right). 
\ee
By Lemma~\ref{lm:cond1OnU},  $ \rho_{\omega}=U_{\omega}\rho_{\omega} U_{\omega^\dagger} $ i.e.,
\be
\frac{1}{2} \left(\Tr_{\bar{\omega}}(I)+\Tr_{\bar{\omega}}(m_i^g)\right) 
 = \frac{1}{2} \left(\Tr_{\bar{\omega}}(I)+U_{\omega}\Tr_{\bar{\omega}}(m_i^g)U_{\omega}^\dagger\right)
\ee
and we obtain 
$\Tr_{\bar{\omega}}(m_i^g) = U_{\omega}\Tr_{\bar{\omega}}(m_i^g)U_{\omega}^\dagger$. Since $m_i^g$ has no support outside 
$\omega$ we can conclude that $m_i^g= Um_i^gU^\dagger$. 

We have shown so far that any transversal encoded gate of $Q$ maps every minimal operator whose support
lies in the support of a site or face operator to itself under
conjugation. We now show that this restriction implies that  $U$
must be such that every $U_i$ is a Clifford unitary. 
If $U$ is a non-Clifford encoded gate, then there exists a $j\in\{1,2,\ldots, n \}$ such that $U_j \not\in \mc{K}_1$.

By Lemma~\ref{lm:minimalType}, there exist a site operator $A_v$ and a face operator $B_f$ such that $j\in \supp(A_v)\cap \supp(B_f)$. This implies the existence of an $X$-only minimal operator $g$ and a $Z$-only minimal operator $h$ such that 
$UgU^\dagger = g$ and $UhU^\dagger= h$. Hence, 
\be
U_jg_jU_j^\dagger \propto g_j \mbox{ and } U_j h_jU_j^\dagger \propto  h_j, \mbox{  up to a scalar.}
\ee 
Since $g$ is an $X$-only  operator $g_j=X$ and $h$ a $Z$-only  operator $h_j=Z$. This implies that 
\be
U_jX U_j^\dagger \propto X \mbox{ and } U_j Z U_j^\dagger \propto  Z, \mbox{  up to a scalar.}
\ee 
Thus $U_j \in \mc{K}_1$ contrary to the assumption that it is not in $\mc{K}_1$. Thus there exists no transversal encoded
non-Clifford gate for the surface codes under the assumptions stated.
\end{proof}

\section{Matroids and Surface Code States}\label{sec:matroids}

Matroids are useful mathematical structures that find applications in many areas such as graph theory, optimization, error-correcting codes, cryptography. However, matroids are yet to find comparable applications in quantum information theory.
One of our goals is to characterize surface code states in terms on matroids. 
With the results of Section~\ref{sec:surfaceStates}
in hand we make the connection to matroids by showing how to  associate a matroid to every 
CSS surface code state.  We call 
matroids  arising from CSS surface code states ``surface code matroids". Two key operations on matroids are deletion and 
contraction. We show that the matroids that are obtained by these minor operations are also surface code matroids in that they can be associated to surface code states. The surface code matroids are in this sense minor closed and can be characterized in terms
of a list of excluded minors.  We refer the reader to \cite{oxley04} for an introduction to matroids. 

\subsection{Surface  Code Matroids} Having defined surface code states, we now associate a matroid to a surface code state 
$\ket{\psi_\Gamma}$ in a canonical fashion. Recall that a CSS surface code state is stabilized by 
\be
S(\psi_{\Gamma}) = \spn{A_v, B_f, \overline{X}_1,\ldots, \overline{X}_l,\overline{Z}_{l+1},\ldots,\overline{Z}_k\bigg|
\ba{l} v\in V(\Gamma)\\f\in F(\Gamma)\ea},
\ee
where we assume that the encoded $X$ operators $\overline{X}_i$ are $X$-only operators and the encoded $Z$ operators
$\overline{Z}_i$ are $Z$-only operators.
In matrix form $S=\left[ \ba{c|c} S_X& 0 \\ 0 & S_Z \ea\right]$.
Since a CSS stabilizer state is completely determined by  $S_X$ or $S_Z$, we can define the CSS
surface code states in terms of $\{ A_v \mid v\in V(\Gamma) \}\cup\{\overline{X}_1,\ldots, \overline{X}_l\}$ or $\{B_f\mid f\in F(\Gamma)\}\cup\{ \overline{Z}_{l+1},\ldots, \overline{Z}_k\}$. 
 Let the vertex-edge incidence matrix of $\Gamma$ be $\mathbb{I}_{V(\Gamma)}$. 
Denote by $\mathfrak{C}(\Gamma) $ the cycles of $\Gamma$. If $B\subseteq \mathfrak{C}(\Gamma)$,
the we denote its  edge incidence matrix by $\mathbb{I}_{B}$.
The supports of the encoded $X$ and encoded 
$Z$ operators are cycles in $\Gamma^\ast$ and $\Gamma$ respectively. 
Denote  by $C(\Gamma^\ast)$, the cycles in $\Gamma^\ast$ that correspond to $\{\overline{X}_1,\ldots, \overline{X}_l\}$.
We can write $S_X$ in terms of these incidence matrices as 
\ben
S_X=\left[\ba{c}\mathbb{I}_{C(\Gamma^\ast) }  \\ \mathbb{I}_{V(\Gamma)}  \ea \right],
\een
where $C(\Gamma^\ast) \subseteq \mathfrak{C}    (\Gamma^\ast) $.
The surface code matroid of $\ket{\psi_\Gamma}$ is defined as the vector matroid of $S_X$ and we shall denote it as $\mc{M}(\psi_\Gamma)$.  
In other words,   $\mc{M}(\psi_\Gamma)$ is determined by all the trivial cycles of $\Gamma^\ast$ and a subset of the nontrivial cycles
of $\Gamma^\ast$. Note that $\mathbb{I}_{V(\Gamma)}$ contains some dependent rows.
However, the matroid associated does not change when we add dependent rows to the matrix representing the matroid. 

The vector matroid associated to $S_Z$ is the dual matroid of the vector  matroid of $S_X$. Since $S_Z$ is determined by 
$\{B_f\mid\in F(\Gamma)\}\cup \{ \overline{Z}_{l+1},\ldots, \overline{Z}_k\}$, by duality we see that 
\ben
\mc{M}(\psi_\Gamma)^\ast  &= & \left[ \ba{c}\mathbb{I}_{C'(\Gamma) }  \\ \mathbb{I}_{V(\Gamma^\ast)} \ea \right],
\een
where $C'(\Gamma)$ is a subset of $\mathfrak{C}(\Gamma)$.
 
\begin{lemma}\label{lm:minorSurfMatroid}
Let $\mc{M}(\psi_\Gamma)$ be a surface code matroid with ground set $E(\Gamma)$. Then any minor of $\mc{M}(\psi_\Gamma)$  is also a surface code matroid. Furthermore, 
\begin{compactenum}[a.]
\item $\mc{M}(\psi_\Gamma)\setminus e = \mc{M}(\psi_{\Gamma\setminus e})$ 
\item $\mc{M}(\psi_\Gamma)/e = \mc{M}(\psi_{\Gamma/e}) $
\end{compactenum}
where $\ket{\psi_{\Gamma\setminus e}}$ and $\psi_{\Gamma/e}$ are some surface code states of $\Gamma\setminus e $ and $\Gamma/e$ respectively.
\end{lemma}
\begin{proof}
Since $M$ is a surface code matroid, there exists a graph  $\Gamma$ embedded on some surface $\Sigma$, such that the surface code 
matroid $\mc{M}(\psi_\Gamma)$ can be represented as 
\be
\mc{M}(\psi_\Gamma)  =\left[ \ba{c}\mathbb{I}_{C(\Gamma^\ast) }  \\ \mathbb{I}_{V(\Gamma)} \ea \right],
\ee
where $C(\Gamma^\ast)$ is some subset of the nontrivial (homological) cycles of $\mathfrak{C}(\Gamma^\ast)$. 
We shall show that both the matroid deletion and contraction operations on $\mc{M}(\psi_\Gamma)$ result in surface code 
matroids. First let us consider the deletion operation. This corresponds to the deletion of a column of $\mc{M}(\psi_\Gamma)$. 
We show that this is equivalent to the deletion of the edge associated to that column. Without loss of generality assume that the
first column is deleted.  Assume that this column is associated to the edge  $e$ in $\Gamma$. 
 Deleting the column corresponding to $e$, the matrix $I_{V(\Gamma)}$ gives the incidence matrix of 
 $\mathbb{I}_{V(\Gamma\setminus e)}$. Deleting an edge in $\Gamma$  corresponds to contracting an edge in $\Gamma^\ast$, see
 Lemma~\ref{lm:graphOps}.
 So any cycle $c\in C(\Gamma^\ast)$ that contains $e$ continues to be a cycle in $\Gamma^\ast/e$ unless $c$ is a coloop. 
If $c$ is a coloop, then contracting $e$ removes the row corresponding to $c$ in $I_{C(\Gamma^\ast)}$
 The cycles in $C(\Gamma^\ast)$ are also in $\Gamma^\ast/e$ which is obtained by contracting the 
edge $e$ in $\Gamma^\ast$. Denote this subset of cycles in $\Gamma^\ast/ e$ by $C(\Gamma^\ast/e)$.
In either case deleting the column in $I_{C(\Gamma^\ast)}$ gives a matrix which is the incidence matrix of cycles in 
$C(\Gamma^\ast/e)$. 
Since $\Gamma^\ast/e = (\Gamma\setminus e)^\ast$,
$\mc{M}(\psi_\Gamma)\setminus e$ can be identified with 
 \be
\mc{M}(\psi_\Gamma)\setminus e  =\left[ \ba{c}\mathbb{I}_{C(\Gamma^\ast/e) }  \\ \mathbb{I}_{V(\Gamma\setminus e)} \ea \right]  
=\left[ \ba{c}\mathbb{I}_{C((\Gamma\setminus e)^\ast) }  \\ \mathbb{I}_{V(\Gamma\setminus e)} \ea \right] 
=\mc{M}(\psi_{\Gamma\setminus e}),
\ee
which shows that the $\mc{M}(\psi_\Gamma)\setminus e$ is a surface code matroid. 

Next let us consider the matroid contraction of $\mc{M}(\psi_{\Gamma})$. This corresponds to a projection of the matrix
representing $\mc{M}(\psi_{\Gamma})$ by removing a column as well as a row.
If $e$ is not a loop, then it is incident on two vertices $u$
and $v$ and the incidence matrix $\mathbb{I}_{V(\Gamma)}$ contains precisely two rows $r_u$ and $r_v$
that have a `1' in the column corresponding to $e$.  We can replace one of the rows say $r_v$ by their sum $r_u+r_v$. 
Then this row corresponds to the incidence vector of the vertex obtained by
contracting along $e$.  The remaining row $r_u$ corresponds a cycle $c$ in $\Gamma^\ast$ . Suppose there is a 
cycle $c'$ in $C(\Gamma^\ast)$ such that it contains $e$, then this cycle can be replaced by another cycle $c''$ that is obtained by the combination of $c$ and $c'$. This does not affect $\mc{M}(\psi_\Gamma)$, therefore we can assume that all the cycles in $C(\Gamma^\ast)$ do not contain $e$. Hence, the
cycles in $C(\Gamma^\ast)$ are also in $\Gamma^\ast \setminus e$ which is obtained by deleting the 
edge $e$ in $\Gamma^\ast$. The matroid minor $\mc{M}(\psi_\Gamma)/e$ is obtained by removing the row $r_u$ in $\mc{M}(\psi_\Gamma)$ and deleting the column corresponding to $e$, which is now all zero except in the row $r_u$, because of the elimination operations performed earlier.

Suppose that $e$ is a loop, then the column corresponding to $e$ in $\mathbb{I}_{V(\Gamma)}$ is an all zero column.
If all the cycles in $C(\Gamma^\ast)$ do not contain $e$, then this column is all zero column and we can simply delete it.
Therefore, we have $\mc{M}(\psi_\Gamma) / e = \mc{M}(\psi_{\Gamma\setminus e}) = \mc{M}(\psi_{\Gamma/ e})  $,
where we used the fact that $\Gamma\setminus e =\Gamma/e $, when $e$ is a loop. On the other hand if some cycle 
$c\in C(\Gamma^\ast)$ contains $e$, then we replace every other cycle $c'$ in $C(\Gamma^\ast)$ that contains $e$
by another cycle $c'' $ such that the row space of $\mathbb{I}_{C(\Gamma^\ast)} $  does not change. At this point only one cycle in $C(\Gamma^\ast)$  contains $e$ . 
Denote these cycles that do not contain $e$ as $C(\Gamma^\ast\setminus e )$. 
We obtain $\mc{M}(\psi_\Gamma)/e$ by removing the row 
corresponding to $c$ and deleting the column corresponding to $e$. 

Whether $e$ is loop or not, on contracting the edge $e$,   $\mathbb{I}_{V(\Gamma)}$  gives $\mathbb{I}_{V(\Gamma/ e)}$, the vertex-edge incidence matrix of $\Gamma/e$, while $\mathbb{I}_{C(\Gamma^\ast)}$ gives 
$\mathbb{I}_{C(\Gamma^\ast\setminus e)}$, cycle-edge incidence matrix of $C(\Gamma^\ast\setminus e)$. Thus the  matroid minor $\mc{M}(\psi_\Gamma)/e$ is given by 
\be
\mc{M}(\psi_\Gamma)/ e  =\left[ \ba{c}\mathbb{I}_{C(\Gamma^\ast\setminus e) }  \\ \mathbb{I}_{V(\Gamma / e)} \ea \right]  
=\left[ \ba{c}\mathbb{I}_{C((\Gamma/ e)^\ast) }  \\ \mathbb{I}_{V(\Gamma/ e)} \ea \right] 
=\mc{M}(\psi_{\Gamma/ e}).
\ee
This completes the proof that the minor of a surface code matroid is also a surface code matroid. 
\end{proof}

A family of matroids $\mc{F}$ is said to be minor closed, if every minor of a matroid in $\mc{F}$  is also in $\mc{F}$. The preceding Lemma, therefore gives us the following result.
\begin{corollary}\label{co:surfMat}
Surface code matroids are minor closed family of matroids. 
\end{corollary}

Corollary~\ref{co:surfMat} establishes a connection with an important, albeit difficult, problem in matroid theory. A consequence of the structure theory of binary matroids is that any class of minor closed binary matroids is characterized by a finite set of excluded minors \cite{geelen09}. Thus all CSS surface code states can be characterized by a finite set of CSS states. 
However, finding these CSS states or equivalently the excluded minors of the surface code matroids appears to be a difficult task.  Fortunately, for some restricted classes we can do so and derive some useful results.

\subsection{A Class of Surface Code Matroids}
In conjunction with Corollary~\ref{co:surfMat} we can state Theorem~\ref{th:surfaceCode} in a slightly more useful way in that we can eliminate many CSS stabilizer states from  being counterexamples to the LU-LC conjecture. Some of the surface code matroids can be related to well studied classes of
matroids, namely graphic and cographic matroids. 
A binary matroid is graphic if and only if it does not have a minor in the set $ \{F_7, F_7^\ast, \mc{M}^\ast, \mc{M}^\ast(K_5), \mc{M}^\ast(K_{3,3})\}$, where 
$F_7$ and $F_7^\ast$ are the Fano matroid and its dual; $\mc{M}(K_5)$, and $\mc{M}(K_{3,3})$ are the matroids of
the complete graph $K_5$ and complete bipartite graph $K_{3,3}$ while $\mc{M}^\ast(K_5)$, and $\mc{M}^\ast(K_{3,3})$ are their
dual matroids; definitions of these matroids can be found in \cite{oxley92}. 

Let $\Gamma$ be a graph and as before denote by  $\mathfrak{C}(\Gamma)$ the set of all cycles of $\Gamma$ and let
 $\mathbb{I}_{V(\Gamma)}$  be the vertex-edge incidence matrix of $\Gamma$. The rowspace of $I_{V(\Gamma)}$ is called the cut space of $\Gamma$.  The incidence matrix of all the cycles of $\Gamma$  is denoted as 
$\mathfrak{C}(\Gamma)$. The row space of $\mathbb{I}_{\mathfrak{C}(\Gamma)}$ is called the flow space of $\Gamma$. The cut space and the flow space of $\Gamma$ are orthogonal to each other, more precisely $\mathbb{I}_{V(\Gamma)} \mathbb{I}_{\mathfrak{C}(\Gamma)}^t=0$. The supports of site operators form a spanning set for the cut space, while the supports of the face operators and the encoded $Z$-operators are a spanning set for the flow 
space. Since $\mathbb{I}_{V(\Gamma)}$ and $\mathbb{I}_{\mathfrak{C}(\Gamma)}$ are orthogonal we can define a stabilizer code with the stabilizer matrix
\ben
S=\left[\ba{c|c}B(C) &\bf{0}\\\bf{0}& B(F)\ea\right],
\een
where $B(C)$ is a basis of the cut space and $B(F)$ is a basis of the flow space. 
Such stabilizer states are sometimes called cut-flow states. A cut-flow state is precisely the CSS surface code state that is stabilized by 
site operators, the face operators and the encoded $Z$ operators. In other words 
\be
S=\spn{A_v, B_f, \overline{Z}_1, \ldots, \overline{Z}_g \mid v\in V(\Gamma), f\in F(\Gamma) }.
\ee
The key observation is that the vector matroid associated to $\mathbb{I}_{V(\Gamma)}$ is precisely the cycle matroid of $\Gamma$ and denoted as $\mc{M}(\Gamma)$. In other words, $\mc{M}(\Gamma)$ is a surface code matroid that is also a cycle matroid. 
The vector matroid associated to $\mathbb{I}_{\mathfrak{C}(\Gamma)}$ is called the bond matroid of $\Gamma$ and denoted as $\mc{M}^\ast(\Gamma)$. The matroids  $\mc{M}(\Gamma)$ and $\mc{M}^\ast(\Gamma)$ are duals of each other.

Theorem~\ref{th:surfaceCode} implies the following corollary that makes a useful connection to binary matroids. 
\begin{corollary}\label{co:matroidConxn}
Let $C$ be an $[n,k,d\geq 3]_2$  classical code with dual distance $d^\perp \geq 3$.  Let $\ket{C}=\sum_{c\in C} \ket{c}$ be
the CSS stabilizer state derived from $C$ with the stabilizer matrix
$S= \left[\ba{c|c}G&\bf{0}\\\bf{0}& H \ea\right]$, where  $G$  and $H$ are the generator  and parity check matrices of $C$.
If the vector matroid associated to $G$ is either graphic or  cographic, then $\ket{C}$ cannot be a counterexample to the $\lu$-$\lc$ conjecture. 
\end{corollary}
\begin{proof}
We note that the matroids associated to $G$ and $H$ are dual to each other. If the 
vector matroid of  $G$ is graphic, then there exists a graph $\Gamma$ and a surface $\Sigma$ on which $\Gamma$ can be embedded such that the cycle matroid of $\Gamma$ is the vector matroid of $G$. 
Since the bond matroid of $\Gamma$ is the dual of cycle matroid, 
it is given by the vector matroid of $H$. It follows that  $\ket{C}$ is a cut-flow state. 
Since $d\geq 3$ and $d^\perp\geq 3$, it follows that $\Gamma$ does not have cycles or cocycles of length $\leq 2$; thus Theorem~\ref{th:surfaceCode} is applicable to $\ket{C}$.  Therefore, $\lu(\ket{C}) = \lc(\ket{C})$ and $\ket{C}$ 
cannot be a counterexample to the LU-LC conjecture. 

Similarly, if the vector matroid of $G$ is cographic, then vector matroid of $H$ is graphic and  we can conclude that the 
state stabilized by $\left[\ba{c|c}H&\bf{0}\\\bf{0}& G \ea\right]$ cannot be counterexample to the $\lu$-$\lc$ conjecture. 
But this state is local Clifford equivalent to $\ket{C}$, therefore $\ket{C}$ cannot be a counterexample either.
\end{proof}

Corollary~\ref{co:matroidConxn} gives us a very simple test to rule a large class of CSS stabilizer states, in particular all the 
cut-flow states. So to test whether a given code $C$ can give rise to a counterexample we simply have to test if the associated matroid is graphic or cographic. Graphic and cographic matroids are characterized by finite number of excluded minors.   In Corollary~\ref{co:matroidConxn} we test to see if 
 $G$ or $H$ are graphic. This test can be performed in polynomial time for binary matroids, see  \cite{tutte60}.

Further, Corollary~\ref{co:matroidConxn} implies that a CSS counterexample for the LU-LC conjecture must necessarily be induced by a nongraphic and noncographic matroid. 
However, a CSS stabilizer state whose associated matroid is neither graphic nor cographic is not necessarily a
counterexample to the LU-LC conjecture.  An interesting albeit difficult problem would be to find the 
excluded minors of the surface code matroids as it would give a sufficient condition for a general CSS state to be ruled out as a counterexample.  

While these associations with matroids are interesting we note that the matroids whose CSS states are LU-LC equivalent are
not minor closed. However, characterizing the largest class of matroids which are minor closed might be an interesting problem
and provide an insight into those states which preserve this property.

\section{Summary}
In this paper we have 
have given a constructive method to compose new counterexamples for the $\lu$-$\lc$ conjecture.  We have also investigated the equivalence classes of two important classes of stabilizer states--the surface code states and cluster states. This allows us to 
rule out the existence of encoded non-Clifford  transversal gates for surface codes. 
Additionally, we have been able to  make a connection with the theory of binary matroids, opening the possibility to approach this 
topic from a different vantage point. 

\section*{Acknowledgment}
We would like to thank Jim Geelen, and Markus Grassl for useful discussions. Some of the results of this paper have been 
presented at  the Workshop on Applications of Matroid Theory and Combinatorial 
Optimization to Information and Coding Theory, Banff International Research Station, Banff, 2009.
This research is supported by NSERC, CIFAR and 
MITACS.

\bibliographystyle{plain}
\def\cprime{$'$}


\end{document}